\documentclass[letterpaper,11pt]{article}
\usepackage{amsmath, amssymb}
\usepackage{color}
\usepackage{fullpage}
\usepackage[numbers]{natbib}
\usepackage[noend]{algorithmic}
\usepackage{tikz}
\usepackage{enumerate}
\usepackage{graphicx}
\usepackage{caption}
\usepackage{hyperref}
\usepackage{subcaption}
\usepackage{zerosum,csquotes}
\usepackage{enumitem,linegoal}
\usepackage[linesnumbered,ruled,vlined]{algorithm2e}
\usepackage{comment}	
\usepackage{ctable}					
\newtheorem{observation}{Observation}[section]
\usepackage{mathtools}
\usepackage[flushleft]{threeparttable}
\usepackage[normalem]{ulem}

\usepackage{footnote}
\makesavenoteenv{tabular}
\makesavenoteenv{table}

\usepackage[margin=1in]{geometry}

 \newtheorem{theorem}{Theorem}[section]
 \newtheorem{lemma}[theorem]{Lemma}

\makeatletter
\def\GrabProofArgument[#1]{ #1: \egroup\ignorespaces}
\def\proof{\noindent\textbf\bgroup Proof%
	\@ifnextchar[{\GrabProofArgument}{. \egroup\ignorespaces}}

\makeatother

\newcommand*\samethanks[1][\value{footnote}]{\footnotemark[#1]}

\newcounter{proccnt}

\newcommand{\konote}[1]{}

\title{Stochastic k-Server: How Should Uber Work?*}

\author{
	Sina Dehghani \thanks{University of Maryland}
	\thanks{Supported in part by NSF CAREER award CCF-1053605,  NSF BIGDATA grant IIS-1546108, NSF AF:Medium grant CCF-1161365, DARPA GRAPHS/AFOSR grant FA9550-12-1-0423, and another DARPA SIMPLEX grant.}
	\and Soheil Ehsani \samethanks[1] \samethanks[2]
	\and MohammadTaghi HajiAghayi \samethanks[1] \samethanks[2]
	\and Vahid Liaghat \thanks{Facebook}
	\and Saeed Seddighin \samethanks[1] \samethanks[2]
}

\begin{document}
	\newcommand{\ignore}[1]{}
\renewcommand{\theenumi}{(\roman{enumi}).}
\renewcommand{\labelenumi}{\theenumi}
\sloppy

%
%

\date{}
\newcommand{\cpx}[1]{\mathcal{O}(#1)}
\newcommand{\metric}[0]{\mathcal{M}}
\newcommand{\HST}[0]{\mathcal{H}}
\newcommand{\E}[0]{\mathbb{E}}
\newcommand{\R}[0]{\mathcal{R}}
\newcommand{\calX}[0]{\mathcal{X}}
\newcommand{\calZ}[0]{\mathcal{Z}}
\newcommand{\Z}[0]{Z}
\newcommand{\I}{\ensuremath{\operatorname{I}}}
\newcommand{\generate}[0]{\ensuremath{\operatorname{generate}}}
\newcommand{\req}{\rho}
\newcommand{\KO}{ \ensuremath{\operatorname{OFKS}}}
\newcommand{\OPT}{ \ensuremath{\operatorname{OPT}}}
\newcommand{\A}{\mathcal{A}}
\newcommand{\B}{\mathcal{B}}
\newcommand{\Line}{\mathcal{L}}

\newcommand{\ALG}{\ensuremath{\operatorname{ALG}}}
\newcommand{\vlnote}[1]{{\color{red}$\ll$\textsf{#1 --VL}$\gg$\marginpar{\color{red}\tiny\bf VL}}}

\newcommand{\Mc}[1]{\ensuremath{\mathcal{#1}}}
\newcommand{\cost}[1]{\ensuremath{\left| #1 \right|}}
\newcommand{\ex}[1]{\ensuremath{\operatorname{E}\left[ #1 \right]}}
\newcommand{\expected}[2]{\ensuremath{\operatorname{E}_{#2}\left[ #1 \right]}}

\newcommand{\afrac}{ALG_{f}}
\newcommand{\aint}{ALG_{i}}

\newcommand{\snote}[1]{{\color{red}$\ll$\textsf{#1 --S}$\gg$\marginpar{\color{red}\tiny\bf S}}}

\newcommand{\floor}[1]{\lfloor #1 \rfloor}
\newcommand{\ceil}[1]{\lceil #1 \rceil}

\maketitle

\thispagestyle{empty}

\begin{abstract}
In this paper we study a stochastic variant of the celebrated $k$-server problem. In the $k$-server problem, we are required to minimize the total movement of $k$ servers that are serving an online sequence of $t$ requests in a metric.
In the stochastic setting we are given $t$ independent distributions $\langle P_1, P_2, \ldots, P_t\rangle$ in advance, and at every time step $i$ a request is drawn from $P_i$. 

Designing the optimal online algorithm in such setting is NP-hard, therefore the emphasis of our work is on designing an approximately optimal online algorithm.
We first show a structural characterization for a certain class of \textit{non-adaptive} online algorithms. We prove that in general metrics, the best of such algorithms has a cost of no worse than three times that of the optimal online algorithm. Next, we present an integer program that finds the optimal algorithm of this class for any arbitrary metric.
Finally by rounding the solution of the linear relaxation of this program, we present an online algorithm for the stochastic $k$-server problem with an approximation factor of $3$ in the line and circle metrics and factor of $O(\log n)$ in a general metric of size $n$. In this way, we achieve an approximation factor that is independent of $k$, the number of servers.

Moreover, we define the {\em Uber} problem, motivated by extraordinary growth of online network transportation services. In the Uber problem, each demand consists of two points -a source and a destination- in the metric. Serving a demand is to move a server to its source and then to its destination. The objective is again minimizing the total movement of the $k$ given servers. We show that given an $\alpha$-approximation algorithm for the $k$-server problem, we can obtain an $(\alpha+2)$-approximation algorithm for the Uber problem. Motivated by the fact that demands are usually highly correlated with the time (e.g. what day of the week or what time of the day the demand has arrived), we study the {\em stochastic Uber} problem. Using our results for stochastic $k$-server we can obtain a 5-approximation algorithm for the stochastic Uber problem in line and circle metrics, and a $O(\log n)$-approximation algorithm for general metrics.

Furthermore, we extend our results to the correlated setting where the probability of a request arriving at a certain point depends not only on the time step but also on the previously arrived requests.

\end{abstract}

\section{Introduction}
The $k$-server problem is one of the most fundamental problems in online computation that has been extensively studied in the past decades. In the $k$-server problem we have $k$ mobile servers on a metric space $\Mc{M}$. We receive an online sequence of $t$ requests where the $i^{th}$ request is a point $r_i\in \Mc{M}$. Upon the arrival of $r_i$, we need to move a server to $r_i$, at a cost equal to the distance from the current position of the server to $r_i$. The goal is to minimize the total cost of serving all requests.

Manasse, McGeoch, and  Sleator~\cite{Manas90} introduced the $k$-server problem as a natural generalization of several online problems, and a building block for other problems such as the metrical task systems. They considered the adversarial model, in which the online algorithm has no knowledge of the future requests. Following the proposition of Sleator and Tarjan~\cite{Sleat85}, they evaluate the performance
of an online algorithm using competitive analysis.
In this model, an online algorithm $\ALG$ is compared to
an \textit{offline} optimum algorithm $\OPT$ which is aware of the entire input in advance. For a sequence of requests $\req$, let $\cost{\ALG(\req)}$ and $\cost{\OPT(\req)}$ denote the total cost of $\ALG$ and $\OPT$ for serving $\req$. An algorithm is \textit{$c$-competitive} if for every $\req$, $\cost{\ALG(\req)} \leq c \cost{\OPT(\req)}+c_0$ where $c_0$ is independent of $\req$.

Manasse \textit{et al.}~\cite{Manas90} showed a lower bound of $k$ for the competitive ratio of any deterministic algorithm in any metric space with at least $k+1$ points. The celebrated \textit{$k$-server conjecture} states that this bound is tight for general metrics. For several years the known upper bounds were all exponential in $k$, until a major breakthrough was achieved by Koutsoupias and Papadimitriou~\cite{KoutP95}, who showed that the so-called \textit{work function algorithm} is $(2k-1)$-competitive. Proving the tight competitive ratio has been the ``holy grail'' of the field in the past two decades. This challenge has led to the study of the problem in special spaces such as the uniform metric (also known as the paging problem), line, circle, and trees metrics (see \cite{chrobak1991new,chrobak1991optimal} and references therein). We also refer the reader to Section~\ref{sec:related} for a short survey of randomized algorithms, particularly the recent result of Bansal, Buchbinder, Madry, and Naor~\cite{bansal2011polylogarithmic} which achieves the competitive ratio of $O(\log^3 n \log^2 k)$ for discrete metrics that comprise $n$ points.

The line metric (or Euclidean 1-dimensional metric space) is of particular interest for developing new ideas. Chrobak, Karloof, Payne, and Vishwnathan~\cite{chrobak1991new} were the first to settle the conjecture in the line by designing an elegant $k$-competitive algorithm. Chrobak and Larmore~\cite{chrobak1991optimal} generalized this approach to tree metrics. Later, Bartal and Koutsoupias~\cite{bartal2004competitive} proved that the work function algorithm is also $k$-competitive in line. Focusing on the special case of $k=2$ in line, Bartal \textit{et al.}~\cite{bartal2000randomized} show that, using randomized algorithms, one can break the barrier of lower bound $k$ by giving a $1.98$-competitive algorithm for the case where we only have two servers.

Despite the strong lower bounds for the $k$-server problem, there are heuristics algorithms that are \textit{constant} competitive in practice. For example, for the paging problem- the special case of uniform metric- the least recently used (LRU) strategy is shown to be experimentally constant competitive
(see Section~\ref{sec:related}). In this paper we present an algorithm an run it on real world data to measure its empirical performance. In particular we use the distribution of car accidents obtained from road safety data. Our experiments illustrate our algorithm is performing even better in practice. 

The idea of comparing the performance of an online
algorithm (with zero-knowledge of the future) to
the request-aware offline optimum has led to crisp
and clean solutions. However, that is not without its downsides. The results in the online model are
often very pessimistic leading to theoretical guarantees that are hardly comparable to experimental results. Indeed, one way to tighten this gap is to use stochastic information about the input data as we describe in this paper.

We should also point out that the competitive analysis is not the only possible or necessarily the
most suitable approach for this problem. Since the distributions from which the input is generated are known,
one can use dynamic programming (or enumeration of future events) to derive the optimal movement of servers. Unfortunately, finding such an optimal online solution using the distributions is an NP-hard problem \footnote{Reduction from $k$-median to Stochastic $k$-server: to find the $k$ median of set $S$ of vertices, one can construct an instance of stochastic $k$-server with $t=1$ and $P_1(v)=1/|S|$ for every $v\in S$. The best initialization of the servers gives the optimum solution to $k$-median of $S$.}, thus the dynamic programming or any other approach takes exponential time. This raises the question that how well one can perform in comparison to the best online solution. In the rest of the paper we formally define the model and address this question.

A natural and well-motivated generalization of $k$-server is to assume the demands are two points instead of just one, consisting of a source and a destination. To serve a demand we need to move a server to the source and then move it to the destination. We call this problem the {\em Uber} problem. One can see, the Uber problem is the same as $k$-server when the sources and the destinations are the same. We also show that, given an $\alpha$-approximation algorithm for the $k$-server problem, we can obtain a $(\alpha+2)$-approximation algorithm for the Uber problem. Thus our results for $k$-server also apply to the Uber problem.

\subsection{The Stochastic Model}
In this paper, we study the \textit{stochastic $k$-server problem} where the input
is not chosen adversarially, but consists of draws
from given probability distributions. This problem has lots of applications such as network transportations and equipment replacement in data centers. The current mega data centers contain hundreds of thousands of servers and switches with limited life-span. For example servers usually retire after at most three years. The only efficient way to scale up the maintenance in data centers is by automation, and robots are designed to handle maintenance tasks such as repairs or manual operations on servers. The replacement process can be modeled as requests that should be satisfied by robots, and robots can be modeled as servers. This problem also has applications in physical networks. As an example, suppose we model a shopping service (e.g. Google Express) as a $k$-server problem in which we receive an online sequence of shopping requests for different stores. We have $k$ shopping cars (i.e., servers) that can serve the requests by traveling to the stores. It is quiet natural to assume that on a certain time of the week/day, the requests arrive from a distribution that can be discovered by analyzing the history. For example, an Uber request is more likely to be from suburb to midtown in the morning, and from midtown to suburb at night.
We formalize this stochastic information as follows.

For every $i\in [1\cdots t]$, a discrete probability distribution $P_i$ is given in advance from which request $r_i$ will be drawn at time step $i$. The distributions are chosen by the adversary and are assumed to be independent but not necessarily identical.
This model is inspired by the well-studied model of \textit{prophet inequalities}
\footnote{In the prophet inequality setting, given (not necessarily identical) distributions $P_1,\ldots,P_t$, an online sequence of values $x_1,\ldots,x_n$ where
	$x_i$ is drawn from $P_i$, an onlooker has to choose one item from the succession
	of the values, where $x_i$ is revealed at step $i$. The onlooker can choose a value
	only at the time of arrival. The goal is to maximize the chosen value.}~\cite{krengel1977semiamarts,hajiaghayi2007automated}.
As mentioned before, the case of line metric has proven to be a very interesting restricted case for studying the $k$-server problem. In this paper, we focus mainly on the class of line metric though our results carry over to circle metric and general metrics as well.

In the adversarial model, the competitive ratio seems to be the only well-defined notion for analyzing the performance of online algorithms. However, in the presence of stochastic information, one can derive a much better benchmark that allows us to make fine-grained distinctions between the online algorithms. We recall that in the offline setting, for a class of algorithms $\Mc{C}$, the natural notion to measure the performance of an algorithm $\ALG\in \Mc{C}$ is the \textit{approximation ratio} defined as the worse case ratio of $\cost{\ALG}$ to $\cost{\OPT(\Mc{C})}$ where $\OPT(\Mc{C})$ is the optimal algorithm in the class. In this paper, we also measure the performance of an online algorithm by its approximation ratio-- compared to the \textit{optimal online solution}. We note that given distributions $P_1,\ldots,P_t$, one can iteratively compute the optimal online solution by solving the following exponential-size dynamic program: for every $i\in [0\cdots t]$ and every possible placement $A$ of $k$ servers (called a \textit{configuration}) on the metric, let $\tau(i,A)$ denote the minimum expected cost of an online algorithm for serving the first $i$ requests and then moving the servers to configuration $A$. Note that $\tau(i,A)$ can inductively be computed via the following recursive formula
\[ \tau(i,A)=\min_B \tau(i-1,B)+ \expected{\text{min. distance from $B$ to $A$ subject to serving $r_i$}}{r_i\sim P_i} \enspace , \]
where $\tau(0,A)$ is initially zero for every $A$.

\subsection{Our Results}
Our first main result is designing a constant approximation algorithm in the line metric when the distributions for different time steps are not necessarily identical.
\begin{theorem}\label{thm:mainline}
	There exists a $3$-approximation online algorithm for the stochastic $k$-server problem in the line metric. The running time is polynomial in $k$ and the sum of the sizes of the supports of input distributions. The same guarantee holds for the circle metric.
\end{theorem}

For the general metric, we present an algorithm with a logarithmic approximation guarantee.

\begin{theorem}\label{thm:maingeneral}
	There exists a $O(\log n)$-approximation online algorithm for the stochastic $k$-server problem in a general metric of size $n$.
\end{theorem}

We prove the theorems using two important structural results. The first key ingredient is a general reduction from class of online algorithms to a restricted class of \textit{non-adaptive algorithms} while losing only a constant factor in the approximation ratio. Recall that a configuration is a placement of $k$-servers on the metric. We say an algorithm $\ALG$ is \textit{non-adaptive} if it follows the following procedure: $\ALG$ pre-computes a sequence of configurations $A_0, A_1,\ldots,A_t$. We start by placing the $k$-servers on $A_0$. Upon the arrival of $r_i$, (i) we move the servers to configuration $A_i$; next (ii) we move the closest server $s$ to $r_i$; and finally (iii) we return $s$ to its original position in $A_i$. We first prove the following structural result.

\begin{theorem}\label{thm:nonadaptive}
	For the stochastic $k$-server problem in the general metric, the optimal non-adaptive online algorithm is within $3$-approximation of the optimal online algorithm.
\end{theorem}

Using the aforementioned reduction, we focus on designing the optimal non-adaptive algorithm. We begin by formulating the problem as an integer program. The second ingredient is to use the relaxation of this program to formalize a natural \textit{fractional variant} of the problem. In this variant, a configuration is a fractional assignment of \textit{server mass} to the points of the metric such that the total mass is $k$. To serve a request at point $r_i$, we need to move some of the mass to have at least one amount of server mass on $r_i$. The cost of moving the server mass is naturally defined as the integral of the movement of infinitesimal pieces of the server mass. By solving the linear relaxation of the integer program, we achieve the optimal fractional non-adaptive algorithm. We finally prove Theorems~\ref{thm:mainline} and~\ref{thm:maingeneral} by leveraging the following rounding techniques. The rounding method in line has been also observed by T{\"u}rkoglu~\cite{turkoglu2005k}. We provide the proof for the case of line in Section~\ref{fred} for the sake of completeness. The rounding method for general metrics is via the well-known embedding of a metric into a distribution of well-separated trees while losing a logarithmic factor in the distortion. Bansal \textit{et al.}~\cite{bansal2011polylogarithmic} use a natural rounding method similar to that of Blum, Burch, and Kalai~\cite{blum1999finely} to show that any fractional $k$-server movement on well-separated trees can be rounded to an integral counterpart by losing only a constant factor.

\begin{theorem}[first proven in \cite{turkoglu2005k}]\label{thm:linefractional}
	Let $\ALG_f$ denote a fractional $k$-server algorithm in the line, or circle. One can use $\ALG_f$ to derive a randomized integral algorithm $\ALG$ such that for every request sequence $\sigma$, $\ex{\cost{\ALG(\sigma)}}=\cost{\ALG_f(\sigma)}$. The expectation is over the internal randomness of $\ALG$. Furthermore, in the stochastic model $\ALG$ can be derandomized.
\end{theorem}

\begin{theorem}[proven in \cite{bansal2011polylogarithmic}]
	Let $\ALG_f$ denote a fractional $k$-server algorithm in any metric. One can use $\ALG_f$ to derive a randomized integral algorithm $\ALG$ such that for every request sequence $\sigma$, $\ex{\cost{\ALG(\sigma)}}\leq O(\log n) \cost{\ALG_f(\sigma)}$.
\end{theorem}

We further show that in the stochastic setting, if the number of possible input scenarios is $m$, even if the distributions are correlated, one can compute the best \textit{fractional} online competitive algorithm in time polynomial in $m$ and $n$. Note that since the number of placements of $k$ servers on $n$ points is exponential, it is not possible to enumerate all the possible choices of an online algorithm. We solve this problem by presenting a non-trivial LP relaxation of the problem with size polynomial in $n$ and $m$; therefore obtaining the following result. We present the formal model and analysis in Appendix~\ref{sec:correlated}.

\begin{theorem}\label{thm:correalted}
	The optimal online algorithm of the stochastic $k$-server problem with correlated setting in line and circle can be computed in polynomial time w.r.t. the number of possible scenarios. In general metrics, an $O(\log n)$-approximation algorithm can be obtained.
\end{theorem} 

We also show that having an $\alpha$-approximation algorithm for $k$-server, we can obtain a $(\alpha+2)$-approximation for the Uber problem, using a simple reduction.

\begin{theorem}
	Let $\ALG$ denote an $\alpha$-approximation algorithm for $k$-server. One can use $\ALG$ to derive a  $(\alpha+2)$-approximation algorithm for the Uber problem.
\end{theorem}
\begin{proof}
	Consider an instance of the Uber problem $I_U$. Let $s_i$ and $t_i$ denote the $i$-th source and destination, respectively. We generate an instance of the $k$-server problem $I_k$ by removing every $t_i$ from $I_U$. In other words the demands are $s_i$'s. We use $\ALG$ to provide a solution for $I_U$ as follows. For satisfying the $i$-th demand, we use $\ALG$ to move a server to $s_i$. Then using the shortest path from $s_i$ to $t_i$, we move that server to $t_i$ and then return it back to $s_i$. Let $\OPT_U$ and $\OPT_k$ denote the cost of the optimal solutions for $I_U$ and $I_k$, respectively. Let $d(s_i, t_i)$ denote the distance of $t_i$ from $s_i$ in the metric. Let $C$ denote the total movement of the servers. We have,
	\begin{align*}
		&\OPT_U \geq \OPT_k.\\
		&\OPT_U \geq \sum_i d(s_i, t_i).\\
		&C \leq \alpha \OPT_k + 2\sum_i d(s_i, t_i) \leq (\alpha+2) \OPT_U.
	\end{align*}

\end{proof}

\subsection{Further Related Work}\label{sec:related}

The randomized algorithms often perform much better in the online paradigm. For the $k$-server problem, a lower bound of $\Omega(\log k)$ is shown by \cite{karloff1994lower} for the competitive ratio of randomized algorithms in most common metrics. Despite the exponential gap, compared to the lower bound of deterministic algorithms, very little is known about the competitiveness of randomized algorithms. In fact, the only known algorithms with competitive ratios below $k$, work either in the uniform metric (also known as the paging problem~\cite{fiat1991competitive,mcgeoch1991strongly,achlioptas2000competitive,bansal2012primal}), a metric comprising $k+1$ points~\cite{fiat2003better}, and two servers on the line~\cite{bartal2000randomized}. Two decades after the introduction of the $k$-server problem, a major breakthrough was achieved by Bansal \textit{et al.} \cite{bansal2011polylogarithmic} in discrete metrics with sub-exponential size. If $\Mc{M}$ comprise $n$ points, their randomized algorithm achieves a competitive ratio of $O(\log^3 n \log^2 k)$.

The case of uniform metric has been extensively studied under various stochastic models motivated by the applications in computer caching.
Koutsoupias and Papadimitriou~\cite{KoutP95} consider two refinements of the competitive analysis for server problems. First, they consider the \textit{diffuse adversary} model. In this model, at every step $i$ the adversary chooses a distribution $D_i$ over the uniform metric of the paging problem.  Then the $i^{th}$ request is drawn from $D_i$ which needs to be served. The distribution $D_i$ is not known to the online algorithm and it may depend on the previous requests. However, in their paper, they consider the case wherein it is guaranteed that for every point $p$, $D_i(p)\leq \epsilon$ for a small enough $\epsilon$; i.e., the next request is not predictable with absolute certainty for the adversary. The results of Koutsoupias and Papadimitriou and later Young~\cite{Young98} shows that the optimum competitive ratio in this setting is close to $1+\Theta(k\epsilon)$.

The second refinement introduced in \cite{KoutP95} restricts the optimal solution to having lookahead at most $\ell$. Hence, one can define a \textit{comparative ratio} which indicates the worst-case ratio of the cost of the best online solution to the best solution with lookahead $\ell$. They show that for the $k$-server problem, and more generally the metrical task system problem, there are online algorithms that admit a comparative ratio of $2\ell+1$; for some instances this ratio is tight.

Various other models of restricting the adversary (access graph model~\cite{borodin1995competitive,irani1996strongly,fiat1997truly}, fault rate model~\cite{karlin2000markov,albers2002paging,denning1983working}, etc) have also been considered for the paging problem (see \cite{Panag06,becchetti2004modeling} and references therein for a further survey of these results).
Unfortunately, many of the stochastic settings considered for the paging problem do not seem to have a natural generalization beyond the uniform metric setting. For example, in the diffuse adversary model, most of the studied distributions do not weaken the adversary in the general metric.
In this paper, we look for polynomial-time \textit{approximation} algorithms in the class of online algorithms that have access to the distributions.

We would like to mention that various online problems have been previously considered under prophet inequality model or i.i.d. model (where all distributions are identical). The maximum matching problem, scheduling, and online network design has been extensively studied in these models(see e.g. \cite{alaei2012online,alaei2011adcell,alaei2013online, dehghani2015online, abolhasani2017beating, dehghani2017online}).
In the graph connectivity problems, Garg, Gupta, Leonardi,and Sankowski~\cite{GargG08} consider the online variants of Steiner tree and several related problems under the i.i.d. stochastic model. In the adversarial model, there exists an $\Omega(\log n)$ lower bound on the competitive ratio of any online algorithm, where $n$ is the number of demands. However, Garg \textit{et al.} show that under the i.i.d. assumption, these problems admit online algorithms with constant or $O(\log \log n)$ competitive ratios. We refer the reader to the excellent book by Borodin and El-Yaniv~\cite{borodin2005online} for further study of online problems.

\section{Preliminaries}
In this section we formally define the stochastic $k$-server problem. The classical $k$-server problem is defined on a metric $\metric$ which consists of points that could be infinitely many. For every two points $x$ and $y$ in metric $\metric$, let $d(x,y)$ denote the distance of $x$ from $y$ which is a symmetric function and satisfies the triangle inequality. More precisely for every three points $x$, $y$, and $z$ we have
\begin{align}
	& d(x,x) = 0 \\
	& d(x,y) = d(y,x)\\
	& d(x,y) + d(y,z) \geq d(x,z).
\end{align}

In the $k$-server problem the goal is to place $k$ servers on $k$ points of the metric, and move these servers to satisfy the requests. We refer to every placement of the servers on the metric points by a \textit{configuration}. Let $\req = \langle r_1,r_2,\ldots,r_t\rangle$ be a sequence of requests, the goal of the $k$-server problem is to find configurations $\langle A_0,A_1,A_2,\ldots,A_t\rangle$ such that for every $i$ there exists a server on point $r_i$ in configuration $A_i$. We say such a list of configurations is \textit{valid} for the given list of requests. A valid sequence of configurations is optimal if $\sum d(A_{i-1},A_i)$ is minimized where $d(X,Y)$ stands for the minimum cost of moving servers from configuration $X$ to configuration $Y$. An optimal sequence $\langle A_0,A_1,\ldots,A_t\rangle$ of configurations is called an optimal offline solution of $\KO(\metric,\req)$ when $\req$ is known in advance. We refer to the optimal cost of such movements with $|\KO(\metric,\req)| = \sum d(A_{i-1},A_i)$. 

We also define the notion of \textit{fractional configuration} as an assignment of the metric points to non-negative real numbers. More precisely, each number specifies a mass of fractional server on a point. Every fractional solution adheres to the following condition: The total sum of the values assigned to all points is exactly equal to $k$. Analogously, a fractional configuration serves a request $r_i$ if there is a mass of size at least 1 of server assigned to point $r_i$. An offline fractional solution of the $k$-server problem for a given sequence of requests $\rho$ is defined as a sequence of fractional configurations $\langle A_0, A_1, \ldots, A_t\rangle$ such that $A_i$ serves $r_i$.

In the online $k$-server problem, however, we are not given the whole sequence of requests in the beginning, but we will be informed of every request once its realization is drawn. An algorithm $\A$ is an online algorithm for the $k$-server problem if it reports a configuration $A_0$ as an initial configuration and upon realization of every request $r_i$ it returns a configuration $A_i$ such that $\langle A_0, A_1, \ldots, A_i\rangle$ is valid for $\langle r_1,r_2,\ldots,r_i\rangle$. If $\A$ is deterministic, it generates a unique sequence of configurations for every sequence of requests. Let $\A(\metric,\req)$ be the sequence that $\A$ generates for requests in $\req$ and $|\A(\metric,\req)|$ denote its cost. 

In the online stochastic $k$-server problem, in addition to metric $\metric$, we are also given  $t$ independent probability distributions $\langle P_1,P_2,\ldots,P_t \rangle$ which show the probability that every request $r_i$ is realized on a point of the metric at each time. An algorithm $\A$ is an online algorithm for such a setting, if it generates a configuration for every request $r_i$ not solely based on $\langle r_1, r_2, \ldots, r_i\rangle$ and $\langle A_0, A_1, \ldots, A_{i-1}\rangle$ but also with respect to the probability distributions. Similarly, we define the cost of an online algorithm $\A$ for a given sequence of requests $\req$ with $|\A(\metric,\req,\langle P_1, P_2, \ldots, P_t\rangle)|$. We define the expected cost of an algorithm $\A$ on metric $\metric$ and with probability distributions $\langle P_1, P_2, \ldots, P_t\rangle$ by 
\begin{equation*}
	|\A(\metric,\langle P_1, P_2, \ldots, P_t\rangle)| = \mathbb{E}_{\forall i, r_i \sim P_i} |\A(\metric,\req,\langle P_1,P_2,\ldots,P_t \rangle)|.
\end{equation*}
For every metric $\metric$ and probability distributions $\langle P_1, P_2, \ldots, P_t\rangle$ we refer to the online algorithm with the minimum expected cost by $\OPT_{\metric,\langle P_1, P_2, \ldots, P_t\rangle}$.

An alternative way to represent a solution of the $k$-server problem is as a vector of configurations $\langle B_0, B_1, \ldots, B_t \rangle$ such that $B_i$ does not necessarily serve request $r_i$. The cost of such solution is equal to $\sum d(B_{i-1},B_i) + \sum 2  d(B_i, r_i)$ where $d(B_i, r_i)$ is the minimum distance of a server in configuration $B_i$ to request $r_i$. The additional cost of $2 d(B_i, r_i)$ can be thought of as moving a server from $B_i$ to serve $r_i$ and returning it back to its original position. Thus, every such representation of a solution can be transformed to the other representation. Similarly, $d(B_i, r_i)$ for a fractional configuration $B_i$ is the minimum cost which is incurred by placing a mass 1 of server at point $r_i$. We use letter $B$ for the configurations of such solutions throughout the paper.

In this paper the emphasis is on the stochastic $k$-server problem on the line metric. We define the line metric $\Line$ as a metric of points from $-\infty$ to $+\infty$ such that the distance of two points $x$ and $y$ is always equal to $|x-y|$. Moreover, we show that deterministic algorithms are as powerful as randomized algorithms in this setting, therefore we only focus on deterministic algorithms in this paper. Thus, from here on, we omit the term deterministic and every time we use the word algorithm we mean a deterministic algorithm unless otherwise is explicitly mentioned.

\section{Structural Characterization}
Recall that an online algorithm $\A$ has to fulfill the task of reporting a configuration $A_i$ upon arrival of request $r_i$ based on $\langle A_0, A_1, \ldots, A_{i-1} \rangle$, $\langle r_1, r_2, \ldots, r_i\rangle$, and $\langle P_1, P_2, \ldots, P_t\rangle$. We say an algorithm $\B$ is \textit{request oblivious}, if it reports  configuration $B_i$ regardless of request $r_i$. As such, $\B$ generates configurations $\langle B_0, B_1, \ldots, B_t\rangle$ for a sequence of requests $\langle r_1, r_2, \ldots, r_t\rangle$ and the cost of such configuration is $\sum d(B_{i-1},B_i) + \sum 2  d(B_i, r_i)$. More precisely, no matter what request $r_i$ is, $\B$ will generate the same configuration for a given list of past configurations $\langle B_0, B_1, \ldots, B_{i-1}\rangle$, a given sequence of past requests $\langle r_1, r_2, \ldots, r_{i-1}\rangle$, and the sequence of probability distributions $\langle P_1, P_2, \ldots, P_t\rangle$. In the following we show that every online algorithm $\A$ can turn into a request oblivious algorithm $\B_{\A}$ that has a cost of at most $|3\A(\metric,\req,\langle P_1, P_2, \ldots, P_t\rangle)|$ for a given sequence of requests $\req$.
\begin{lemma}\label{lm1}
	Let $\A$ be an online algorithm for the stochastic $k$-server problem. For any metric $\metric$, there exists a request oblivious algorithm $\B_{\A}$ such that $$|\B_{\A}(\metric,\langle P_1, P_2, \ldots, P_t\rangle)| \leq 3|\A(\metric,\langle P_1, P_2, \ldots, P_t\rangle)|.$$
\end{lemma}
\begin{proof}
	Let $\req$ be a sequence of requests. We define online algorithm $\B_{\A}$ as follows: The configuration that $\B_{\A}$ reports for a given list of input arguments $\langle B_0, B_1, \ldots, B_i \rangle$, $\langle r_1, r_2, \ldots, r_i\rangle$, and $\langle P_1, P_2, \ldots, P_t\rangle$ is the output of algorithm $\A$ on inputs $\langle B_0, B_1, \ldots, B_i \rangle$, $\langle r_1, r_2, \ldots, r_{i-1}\rangle$, and $\langle P_1, P_2, \ldots, P_t\rangle$ (The same input except that $r_i$ is dropped from the sequence of requests). We show the cost of such algorithm for input $\req$ is at most 3 times the cost of $\A$ for the same input.
	
	Let $\langle A_0, A_1, \ldots, A_t\rangle$ be the sequence of configurations that $\A$ generates for requests $\req$ and $\langle B_0, B_1, \ldots, B_t\rangle$ be the output of algorithm $\B_{A}$. According to the construction of $\B_{\A}$, $B_0 = A_0$ and $B_i = A_{i-1}$ for all $1 \leq i \leq t$. Note that for algorithm $\A$, we assume every $A_i$ serves request $r_i$. By definition, the cost of solution $\langle B_0, B_1, B_2, \ldots, B_t\rangle$ is equal to $\sum d(B_{i-1},B_i) + 2\sum d(B_i, r_i)$. Since $B_0 = B_1 = A_0$ and $B_i$ = $A_{i-1}$, 
	\begin{equation}\label{cc1}
		\sum_{i=1}^{t} d(B_{i-1},B_i) = \sum_{i=1}^{t-1} d(A_{i-1},A_i) \leq \sum_{i=1}^t d(A_{i-1},A_i) = |\A(\metric, \req, \langle P_1, P_2, \ldots, P_t\rangle)|.
	\end{equation}
	Moreover, since every $A_i$ servers request $r_i$, $d(B_i,r_i) \leq d(B_i, A_i) = d(A_{i-1},A_i)$. Hence,
	\begin{equation}\label{cc2}
		2\sum_{i=1}^t d(B_i, r_i) \leq 2\sum_{i=1}^t d(B_i, A_i) = 2\sum_{i=1}^t d(A_{i-1},A_i) = 2|\A(\metric,\req,\langle P_1, P_2, \ldots, P_t\rangle)|.
	\end{equation}
	Inequality \eqref{cc1} along with Equation \eqref{cc2} implies $$|\B_\A(\metric,\req,\langle P_1, P_2, \ldots, P_t\rangle)| \leq 3|\A(\metric,\req,\langle P_1, P_2, \ldots, P_t\rangle)|.$$
	Since this holds for all requests $\req \sim \langle P_1, P_2, \ldots, P_t\rangle$, we have $$|\B_\A(\metric,\langle P_1, P_2, \ldots, P_t\rangle)| \leq 3|\A(\metric, \langle P_1, P_2, \ldots, P_t\rangle)|$$ and the proof is complete.
\end{proof}

An immediate corollary of Lemma \ref{lm1} is that the optimal request oblivious algorithm has a cost of at most $|3\OPT_{\metric,\langle P_1, P_2, \ldots, P_t\rangle}(\metric, \langle P_1, P_2, \ldots, P_t\rangle)|$. Therefore, if we only focus on the request oblivious algorithms, we only lose a factor of 3 in comparison to the optimal online algorithm. The following lemma states a key structural lemma for an optimal request oblivious algorithm.

\begin{lemma}\label{mohem}
	For every request oblivious algorithm $\B$, there exists a randomized request oblivious algorithm $\B'$ with the same expected cost which is not only oblivious to the last request, but also oblivious to all requests that have come prior to this.
\end{lemma}
\begin{proof}
	For any given request oblivious online algorithm $\B$, we construct an online algorithm $\B'$ which is oblivious to all of the requests as follows: For an input $\langle B_1, B_2, \ldots, B_{i-1}\rangle$ of configurations and probability distributions $\langle P_1, P_2, \ldots, P_t\rangle$, draw a sequence of requests $\langle r_1, r_2, \ldots, r_i\rangle$ from $\langle P_1, P_2, \ldots, P_t\rangle$ conditioned on the constraint that $\B$ would generate configurations $\langle B_1, B_2, \ldots, B_{i-1}\rangle$ for requests $\langle r_1, r_2, \ldots, r_{i-1}\rangle$. Now, report the output of $\B$ for inputs $\langle B_1, B_2, \ldots, B_{i-1}\rangle$, $\langle r_1, r_2, \ldots, r_i\rangle$, and $\langle P_1, P_2, \ldots, P_t\rangle$.
	
	We define the cost of step $i$ of algorithm $B'$ as $d(B_{i-1},B_{i}) + 2d(B_i,r_i)$. Due to the construction of algorithm $\B'$, the expected cost of this algorithm at every step $i$ for a random sequence of requests is equal to the expected cost of algorithm $\B$ for a random sequence of requests drawn from $\langle P_1, P_2, \ldots, P_t\rangle$. Therefore, the expected cost of both algorithms for a random sequence of requests are equal and thus $|\B(\metric,\langle P_1, P_2, \ldots, P_t\rangle)| = |\B'(\metric,\langle P_1, P_2, \ldots, P_t\rangle)|$.
\end{proof}

Lemma \ref{mohem} states that there always exists an optimal randomized request oblivious online algorithm that returns the configurations regardless of the requests. We call such an algorithm \textit{non-adaptive}. Since a non-adaptive algorithm is indifferent to the sequence of the requests, we can assume it always generates a sequence of configurations just based on the distributions. For an optimal of such algorithms, all such sequence of configurations should be optimal as well. Therefore, there always exists an optimal non-adaptive online algorithm which is deterministic. By Lemma \ref{lm1} not only do we know the optimal request oblivious algorithm is at most 3-approximation, but also the same holds for the optimal non-adaptive algorithm. 
\begin{theorem}\label{thebest}
	There exists a sequence of configurations $\langle B_0, B_1, \ldots, B_t\rangle$ such that an online algorithm which starts with $B_0$ and always returns configuration $B_i$ upon arrival of request $r_i$ has an opproximation factor of at most 3. 
\end{theorem}

\section{Fractional Solutions}\label{fsf}

In this section we provide a fractional online algorithm for the $k$-server problem that can be implemented in polynomial time. Note that by Theorem \ref{thebest} we know that there exist configurations $\langle \B_1, \B_2, \ldots, \B_t\rangle$ such that the expected cost of a non-adaptive algorithm that always returns these configurations is at most 3 times the cost of an optimal online algorithm. Therefore, we write an integer program to find such configurations with the least expected cost. Next, we provide a relaxed LP of the integer program and show that every feasible solution of such LP corresponds to a fractional online algorithm for the stochastic $k$-server problem. Hence, solving such a linear program, that can be done in polynomial time, gives us a fractional online algorithm for the problem.


\subsection{Linear Program}
Recall that given $t$ independent distributions $\langle P_1, \ldots, P_t \rangle$ for online stochastic $k$-server, an adaptive algorithm can be represented by $t+1$ configurations $\langle B_0, \ldots, B_t \rangle$. Upon the arrival of each request $r_i$, we move the servers from configuration $B_{i-1}$ to $B_i$ and then one server serves $r_i$ and goes back to its position in $B_i$. The objective is to find the configurations such that the cost of moving to new configurations in addition to the expected cost of serving the requests is minimized. Therefore the problem can formulated in an offline manner. First we provide an integer program in order to find a vector of configurations with the least cost.

The decision variables of the program represent the configurations, the movement of servers from one configuration to another, and the way that each possible request is served. In particular,
at each time step $\tau$:
\begin{itemize}
	\item For each node $v$ there is a variable $b_{\tau, v} \in N$ denoting the number of servers on node $v$. 
	\item For each pair of nodes $u$ and $v$, there is a movement variable $f_{\tau, u, v} \in N$ denoting the number of servers going from $u$ to $v$ for the next round.
	\item For each node $v$ and possible request node $r$, there is a variable $x_{\tau, v, r} \in \{0, 1\}$ denoting whether $r$ is served by $v$ or not.
\end{itemize}

In the following integer program, the first set of constraints ensures the number of servers on nodes at each time is updated correctly according to the movement variables. The second set of constraints ensures that each possible request is served by at least one server. The third set of constraints ensures that no possible request is served by an empty node. By the definition, the cost of a sequence of configurations $\langle B_0, \ldots, B_t \rangle$ is $\sum_{i=1}^t d(B_{i-1}, B_i)+2 \sum_{i=1}^t d(B_i, r_i)$. Thus the objective is to minimize the expression $$\sum_\tau\sum_{u, v}f_{\tau, u, v}d(u, v) + 
2 \sum_\tau\sum_v\sum_r x_{\tau, v, r}\Pr(z \sim P_\tau=r)d(v, r)$$, where $\Pr(z \sim P_\tau=r)$ denotes the probability that $r$ is requested at time $\tau$.
\begin{equation}
	\begin{aligned}\nonumber
		& \text{min.} \hspace{1cm}
		&& \sum_\tau\sum_{u, v}f_{\tau, u, v}d(u, v) + 
		2 \sum_\tau\sum_v\sum_r x_{\tau, v, r}\Pr(z \sim P_\tau=r)d(v, r)\\
		&\forall \tau, v	&& b_{\tau+1, v}=b_{\tau, v} + \sum_u f_{\tau, u, v} - \sum_u 
		f_{\tau, v, u}.\\
		&\forall \tau, u, v	&& \sum_v x_{\tau, v, r} \geq 1.\\
		&\forall \tau, v, r	&& x_{\tau, v, r} \leq b_{\tau, v}.\\
		&\forall \tau	&& \sum_v b_{\tau, v} \leq k.\\
		&\forall \tau, v, r	&& x_{\tau, v, r} \in \{0,1\}.\\
		&\forall \tau, u, v	&& f_{\tau, u, v} \in N.\\ 
		&\forall \tau, v	&& b_{\tau, v} \in N.\\ 
	\end{aligned}
\end{equation}

Now we consider the following relaxation of the above integer program.

\begin{equation}
	\begin{aligned}\nonumber
		& \text{min.} \hspace{1cm}
		&& \sum_\tau\sum_{u, v}f_{\tau, u, v}d(u, v) + 
		2 \sum_\tau\sum_v\sum_r x_{\tau, v, r}\Pr(z \sim P_\tau=r)d(v, r)\\
		&\forall \tau, v	&& b_{\tau+1, v}=b_{\tau, v} + \sum_u f_{\tau, u, v} - \sum_u 
		f_{\tau, v, u}.\\
		&\forall \tau, u, v	&& \sum_v x_{\tau, v, r} \geq 1.\\
		&\forall \tau, v, r	&& x_{\tau, v, r} \leq b_{\tau, v}.\\
		&\forall \tau	&& \sum_v b_{\tau, v} \leq k.
	\end{aligned}
\end{equation}

\section{Reduction from Integral $k$-server to Fractional $k$-server}\label{fred}
In this section we show how we can obtain an integral algorithm for the stochastic $k$-server problem from a fractional algorithm. We first show that every fractional algorithm for the line metric can be modified to an integral algorithm with the same cost. Next, we study the problem on HST metrics; we give a rounding method that produces an integral algorithm from a fractional algorithm while losing a constant factor. Finally, we leverage the previously known embedding techniques to show every metric can be embedded into HST's with a distortion of at most $O(\log n)$. This will lead to a rounding method for obtaining an integral algorithm from every fractional algorithm on general metrics while losing a factor of at most $O(\log n)$. Combining this with the $3$ approximation fractional algorithm that we provide in Section \ref{fsf}, we achieve an $O(\log n)$ approximation algorithm for the stochastic $k$-server problem on general graphs. 

\subsection{Integrals Are as Strong as Fractionals On the Line}
In this section we show every fractional algorithm on the line metric can be derandomized to an integral solution with the same expected cost. The rounding method is as follows: For every fractional configuration $A$, we provide an integral configuration $\I(A)$ such that (i) the distance of two configurations $A_1$ and $A_2$ is equal to the expected distance of two configurations $\I(A_1)$ and $\I(A_2)$. (ii) for every point $x$ in the metric that $A$ has a server mass of size at least $1$ on $x$, there exists a server on point $x$ in $\I(A)$.

Let for every point $x$ in the metric, $A(v)$ denote the amount of server mass on node $v$ of the line. For every fractional configuration $B$, we define a mass function $f_A: (0,k] \to V$ as follows.
$f_A(x)=v_j$ if and only if $j$ is the minimum integer such that $\sum_{i=1}^{j-1}A(i) < x$ and $\sum_{i=1}^{j}A(i) \geq x$.
Intuitively, if one gathers the server mass by sweeping the line from left to right, $f_A(x)$ is the first position on which we have gathered $x$ amount of server mass. The rounding algorithm is as follows:
\begin{itemize}
	\item Pick a random real number $r$ in the interval $[0, 1)$.
	\item $\I(A)$ contains $k$ servers on positions $f_A(r)$, $f_A(r+1)$, $f_A(r+2)$, \ldots, $f_A(r+k-1)$. 
\end{itemize}

Note that the rounding method uses the same $r$ for all of the configurations. More precisely, we draw $r$ from $[0, 1)$ at first and use this number to construct the integral configurations from fractional configurations. The following two lemmas show that both of the properties hold for the rounding algorithm we proposed.
\begin{lemma}
	Let $A$ be a fractional configuration and $x$ be a point such that $A(x) \geq 1$. Then $\I(A)$ has a server on $x$.
\end{lemma}
\begin{proof}
	Due to the construction of our rounding method, for every two consecutive servers $a$ and $b$ in $\I(\A)$, the total mass of servers after $a$ and before $b$ in the fractional solution is less than $1$. Therefore, $\I(A)$ should put a server on point $x$, otherwise the total mass of servers in the fractional solution between the first server before $x$ and the first server after $x$ would be at least $1$. 
\end{proof}

The next lemma shows that the rounding preserves the distances between the configurations in expectation.
\begin{lemma}
	Let $A_1$ and $A_2$ be two fractional configurations and $|A_1 - A_2|$ be their distance. The following holds for the distances of the configurations
	$$\mathbb{E}|\I(A_1) - \I(A_2)| = |A_1 - A_2|.$$
\end{lemma}
\begin{proof}
	The key point behind the proof of this lemma is that the distance of two fractional configurations $A_1$ and $A_2$ can be formulated as follows
	\begin{equation*}
		|A_1 - A_2| = \int_0^1 |\I_\omega(A_1) - \I_\omega(A_2)|d_\omega
	\end{equation*}
	where $\I_\omega(A)$ stands for an integral configurations which places the servers on points $f_A(\omega)$, $f_A(\omega+1)$, $f_A(\omega+2)$, $\ldots$, $f_A(\omega+k-1)$. Since at the beginning of the rounding method we draw $r$ uniformly at random, the expected distance of the two rounded configurations is exactly equal to 
	\begin{equation*}
		\int_0^1 |\I_\omega(A_1) - \I_\omega(A_2)|d_\omega
	\end{equation*}
	which is equal to the distance of $A_1$ from $A_2$.
\end{proof}
\begin{theorem}
	For any given fractional online algorithm $\A$ for the $k$-server problem on the line metric, there exists an online integral solution for the same problem with the same expected cost.  
\end{theorem}

\subsection{Reduction for General Graphs}
An HST is a undirected rooted tree in which every leaf represents a point in the metric and the distance of a pair of points in the metric is equal to the distance of the corresponding leaves in the tree. In an HST, weights of the edges are uniquely determined by the depth of the vertices they connect. More precisely, in a $\sigma$-HST the weight of an edges between a vertex $v$ and its children is equal to $\sigma^{h-d_v}$ where $h$ stands for the height of the tree and $d_v$ denotes the depth of vertex $v$.

Since HSTs are very well structured, designing algorithms on HSTs is relatively easier in comparison to a more complex metric. Therefore, a classic method for alleviating the complexity of the problems is to first embed the metrics into HSTs with a low distortion and then solve the problems on these trees. 

Perhaps the most important property of the HSTs is the following:
\begin{observation}
	For every pair of leaves $u,v \in T$ of an HST, the distance of $u$ and $v$ is uniquely determined by the depth of their deepest common ancestor.
\end{observation}
Note that, the higher the depth of the common ancestor is, the lower the distance of the leaves will be. Therefore, the closest leaves to a leaf $v$ are the ones that share the most common ancestors with $v$.
Bansal \textit{et al.} propose a method for rounding every fractional solution of the $k$-server problem to an integral solution losing at most a constant factor \cite{bansal2011polylogarithmic}.
\begin{theorem} \cite{bansal2011polylogarithmic} Let $T$ be a $\sigma$-HST with $n$ leaves, $\sigma > 5$, and let $A = \langle A_0, A_1, A_2, \ldots, A_t\rangle$ be a sequence of fractional configurations. There is an online procedure that maintains a sequence of randomized k-server configurations $S = \langle S_0, S_1, S_2, \ldots, S_t \rangle$ satisfying the following two properties: 
	\begin{itemize}
		\item At any time $i$, the state $S_i$ is consistent with the fractional state $A_i$.
		\item If the fractional state changes from $x_{i-1}$ to $x_i$ at time $i$, incurring a movement cost of $c_i$, then the state $S_{i-1}$ can be modified to a state $S_i$ while incurring a cost of $O(c_i)$ in expectation.
	\end{itemize}
\end{theorem}

Embedding general metrics into trees and in particular HSTs has been the subject of many studies. The seminal work of Fakcharoenphol \textit{et al.}~\cite{fakcharoenphol2003tight} has shown that any metric can be randomly embedded to $\sigma$-HSTs with distortion $O(\frac{\sigma \log n}{\log \sigma})$. 
\begin{theorem}\cite{fakcharoenphol2003tight}
	There exists a probabilistic method to embed an arbitrary metric $\metric$ into $\sigma$-HSTs with distortion $\frac{\sigma \log n}{\log \sigma}$. 
\end{theorem}
Therefore, to round a fractional solution on a general metric, we first embed it into $6$-HSTs with a distortion of at most $O(\log n)$ and then round the solution while losing only a constant factor. This will give us an integral algorithm that has an expected cost of at most $O(\log n)$ times the optimal.

\begin{theorem}
	For any given fractional online algorithm $\A$ for the $k$-server problem on an arbitrary metric, there exists an online integral solution for the same problem having a cost of no worse that $O(\log n)$ times the cost of $\A$ in expectation.  
\end{theorem}

\section{Acknowledgment}
We would like to thank Shi Li for having helpful discussions.

\bibliographystyle{abbrv}
\bibliography{k-server.bib}
\newpage

\appendix

\section{Correlated Setting}\label{sec:correlated}

In this section, we study the $k$-server problem when the probability distributions are not independent. Recall that in the independent setting the sequence of requests is referred to by $\req=\langle r_1,\ldots,r_t\rangle$. In the correlated model we assume all different possibilities for $\req$ have been given in the form of a set $\R=\{\rho_1,\ldots,\rho_m\}$ of $m$ sequences $\rho_i=\langle r_{i, 1},\ldots, r_{i,t} \rangle$. Moreover, we assume the probability of each scenario $\rho_i$ is denoted by $p_i$ and given in advance. Given the list of different scenarios and probabilities, the goal is to design an online algorithm to serve each request $r_{i,j}$ prior to arrival of the next request such that the overall movement of the servers is minimized.

We model this problem by an integer program. We first write an integer program and show that every solution of this program is uniquely mapped to a deterministic online algorithms for the problem. Moreover, every online algorithm can be mapped to a feasible solution of the program. More precisely, each solution of the program is equivalent to an online algorithm of the problem. Furthermore, we show how to derive an online algorithm from the solution of the integer program. These two imply that the optimal deterministic online algorithm can be obtained from the optimal solution of the program. 


\subsection{Program}

To better convey the idea behind the integer program, we first introduce the tree $T$ which is a trie containing all sequences $\rho_1$ to $\rho_m$. Let us use $w(v)$ to denote the path from the root to a node $v$. With these notations, a node $v\in T$ represents a request which may occur conditioning all requests in $w(v)$ occur beforehand. Besides, every leaf of $T$ uniquely represents one of the $\rho_i$'s. Let us use $l(v)$ to denote the set of those indices $i$ for which $\rho_i$ is a leaf of the subtree of $v$. At each step $t$, only those $\rho_i$'s can be a final option for $R$ that $\langle \rho_{i,1},\ldots,\rho_{i,t} \rangle=\langle r_1,\ldots,r_t\rangle$. Hence, a new request $r_t$ can be informative since we know that none of the $\rho_i$'s in $l(r_{t-1})\backslash l(r_tau)$ will occur anymore. For a node $v$ we define $Pr(v)$ as the probability of all requests in $w(v)$ happening i.e. $Pr(v)=\sum_{i\in l(v)} Pr(R=\rho_i)$.


We extend the tree $T$ by adding $k-1$ additional nodes. As shown in Figure 1, these nodes form a path leading to the root of $T$. These nodes plus the root represent the initial configuration of the $k$ servers. Let us call these nodes the initial set $I$. Now we can show the movement of the servers in our metric space by means of $k$ tokens in $T$. To do so, we begin with putting one token on each of the $k$ nodes of $I$. Each token corresponds to one of the servers. After a server moves to serve a request $r_t$, we move its corresponding token to a node of $T$ which represents the request $r_t$. Note that at this step, there is no discrimination between any of the sequences in $l(r_t)$ in terms of occurrence. This causes a deterministic online algorithm $A$ to serve the first $|w(r_t)|$ requests of $R$ in the same way if $R$ is going to be one of $\rho_i$'s $(i\in l(r_t))$. A result of this uniquely serving is that we can use some downward links on $T$ in order to show how each request $v$ gets served. In the next paragraphs we explain about these links and how we construct the integer program.

Let us use $x_{u,v}$ to denote a link from a node $u\in T$ to its descendant $v$. $x_{u,v}$ is one if and only if $A$ uses the same server to serve $u$ and then $v$ without using that server to serve any other request between $u$ and $v$. This consecutive serving may occur with probability $Pr(v)=Pr(u)Pr(v|u)$. In this case, the algorithm moves a server from $u$ to $v$ and pays $|u-v|$ as the distance cost between the two points of the metric space corresponding to $u$ and $v$.

There are two conditions for these links that we must care about. First, since each request $v$ should be served with a server, at least one of the $x_{u,v}$'s should be one for all $u$ in $w(v)$. Without loss of generality, we assume this is exactly one of them, i.e. there is no need to serve a request with more than one server. Second, after serving a request $u$, a server can go for serving at most one other request. That is, for each $i\in l(u)$, there should be at most one $v\in \rho_i$ such that $x_{u,v}=1$. This condition guarantees that in serving the sequence of requests $R$, a server which serves $r_{t_1}\in R$ has always at most one other request $r_{t_2}\in R$ as the next serving request.

The following integer program maintains both conditions for $x_{u,v}$'s and has the expected overall movement of all servers as the objective function:

\begin{equation}
\begin{aligned}\nonumber
& \text{min.} \hspace{1cm}    && \sum_{u,v\in T;u\in w(v)} Pr(v) |u-v| x_{u,v} \\
&\forall v\in T\backslash I	  && \sum_{u\in w(v)} x_{u,v} = 1. \\
&\forall u\in T, i\in l(u)	  && \sum_{v\in \rho_i} x_{u,v} \leq 1.\\
&\forall u,v \in T, u\in w(v) && 	x_{u,v} \in \{0,1\}\\ 
\end{aligned}
\end{equation}

Next, we can relax the constraints of the program to make it linear. Therefore, instead of assigning either $\{0\}$ or $\{1\}$, to each $x_{u,v}$ we let it be a real number between 0 and 1. Thus, the integer program turns to the following linear program with the same objective function but more relaxed constraints.

\begin{equation}
\begin{aligned}\nonumber
& \text{min.} \hspace{1cm}    && \sum_{u,v\in T;u\in w(v)} Pr(v) |u-v| x_{u,v} \\
&\forall v\in T\backslash I	  && \sum_{u\in w(v)} x_{u,v} = 1. \\
&\forall u\in T, i\in l(u)	  && \sum_{v\in \rho_i} x_{u,v} \leq 1.\\
&\forall u,v \in T, u\in w(v) && 	x_{u,v} \leq 1\\ 
&\forall u,v \in T, u\in w(v) && 	x_{u,v} \geq 0\\ 
\end{aligned}
\end{equation}
Note that every feasible solution of the linear program is corresponding to a fractional solution of the problem. Since the optimal solution of the linear program can be found in polynomial time, using the rounding methods presented in Section \ref{fred} we obtain an optimal online algorithm for the line metric and a $O(\log n)$ approximation algorithm for general metrics as stated in Theorem~\ref{thm:correalted}.

\section{Experimental Results}
The goal of this section is to make an evaluation of our method for the line on a real world data set. The line can be an appropriate model for a plenty of applications. For example, it could be sending road maintenance trucks to different points of a road or sending emergency vehicles to accident scenes along a highway. For this experiment, we take the case of car accidents.

\textbf{Data sets.} We use Road Safety Data\footnote{https://data.gov.uk/dataset/road-accidents-safety-data/} to find the distribution of the accidents along the A1\footnote{https://en.wikipedia.org/wiki/A1\_road\_(Great\_Britain)} road in Great Britain. In 2015, over 1600 accidents occurred on this highway, with an average of 140 accidents per month. We assume a point every 10 miles along the highway. That is 40 points in total. Then we build the distributions with respect to how the accidents are spread over the days of month. In this way, we achieve 30 distributions for 40 points along the line.

\textbf{Algorithms.} We compare the performance of our method to that of the optimum algorithm. To find the optimum solution we use backtracking. The running time of the algorithm is exponential to $k$. However, we use techniques such as \textit{branch and bound} and \textit{exponential dynamic programming} to get a fast implementation.


\textbf{Results.} We run different experiments with $k$ from 2 to 11 on the line and distributions explained above. In previous sections we showed an upper bound of 3 for the approximation factor of our algorithm. Interestingly, in these experiments we can observe a better performance as shown by Figure \ref{fig}. We compare the running time of the algorithms in Table \ref{table}. Note that the size of our LP our method solvers does not vary by $k$. This is in fact the reason behind why its running time remains almost the same. In contrast, the running time of the optimum algorithm grows exponentially.

\begin{figure}[h]{}
	\centering
	\includegraphics[scale=0.5]{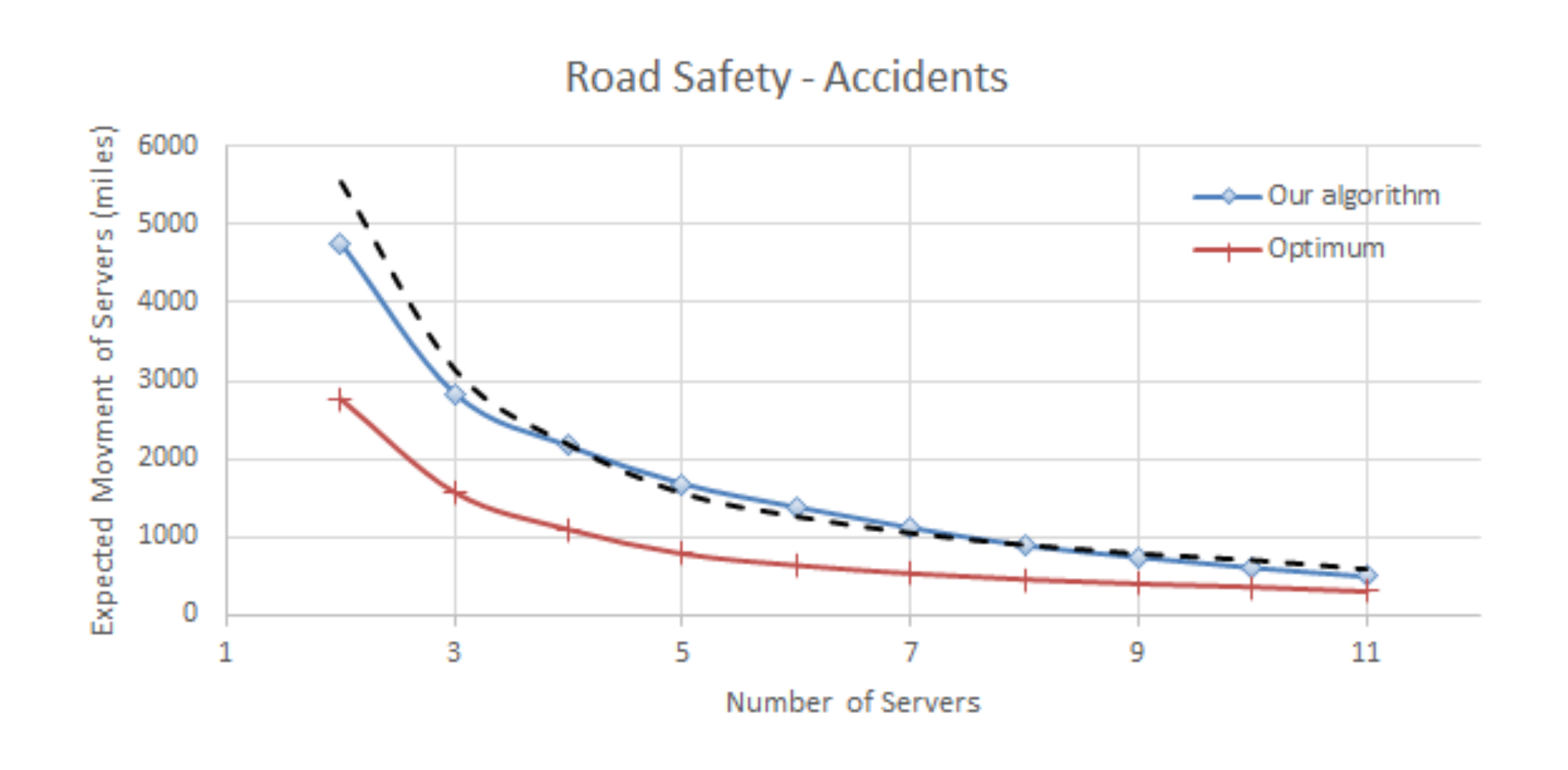}
	\caption{Performance of our algorithm compared to the optimum. The dashed curve indicates two times the optimum.}
	\label{fig}
\end{figure}

\begin{table}
	\begin{tabular}{|l|l|l|l|l|l|l|l|l|l|l|}
		\hline
		Number of Servers & 2 & 3 & 4 & 5 & 6 & 7 & 8 & 9 & 10 & 11 \\ \hline
		Algorithm & 6.5 & 7.6 & 6.7 & 7.1 & 7.5 & 8.3 & 8.5 & 8.4 & 9.3 & 8.2  \\ \hline
		Optimum & 0.2 & 0.8 & 3.1 & 8.4 & 29.4 & 57.9 & 126.3 & 406.7 & 1477.1 & 6173.6  \\ \hline
	\end{tabular}
	\caption{The running time of our algorithm and the optimum algorithm in seconds. For higher number of servers, the optimum solution was not calculable within 5 hours.}
	\label{table}
\end{table}

\end{document}